\documentclass[12pt]{article}

\usepackage{verbatim}
\usepackage{amsmath}
\usepackage{hyperref}
\usepackage{amssymb}
\usepackage{stmaryrd}
\usepackage{tensor}
\usepackage{graphicx}
\usepackage{amsthm}
\usepackage{braket}
\usepackage{accents}
\usepackage{bm}
\usepackage{tikz-cd}
\newcommand{\mbb}{\mathbb}
\newcommand{\ts}{\ \, \;}

\newtheorem{mydef}{Definition}
\newtheorem{mythm}{Theorem}

\usepackage{authblk}
\makeatletter
\renewcommand*\env@matrix[1][\arraystretch]{%
  \edef\arraystretch{#1}%
  \hskip -\arraycolsep
  \let\@ifnextchar\new@ifnextchar
  \array{*\c@MaxMatrixCols c}}
\makeatother

\newenvironment{sciabstract}{%
\begin{quote} \bf}
{\end{quote}}

\title{Causal and self-dual morphisms in four complex dimensions}

\author{Edward B. Baker III\thanks{edwardbaker86@gmail.com}}

\date{\today}

\begin{document} 

\maketitle 

\begin{sciabstract}
We define a class of maps between holomorphically embedded null curves which generalize conformal transformations, and can be defined in any complex dimension.  In four dimensions, we can also define a similar map between self-dual surfaces, which generalize flat $\alpha$-planes.  These maps are respectively called causal and self-dual morphisms.  It is shown that there exist an infinite class of non-trivial examples for both types of maps in four dimensions.  
\end{sciabstract}

\section{Introduction}

In a previous paper\cite{baker2019frustrated}, it was shown that a certain equation in four dimensions, referred to as a frustrated conformal transformation, can generate holomorphic maps of an associated ambitwistor space.  Part of the motivation for these results was to generalize conformal transformations in a way that overcomes the rigidity of these functions in dimension greater than two, which follows from Liouville's theorem of conformal mappings.  Although a solution to this equation based on the BPST instanton was found, further investigation suggests that there may not be more general solutions, and therefore the equations may not have the flexibility desired from the initial motivation of the theory.

In this paper, we will take a different but related approach to the problem.  In many ways this approach is more straightforward, relying less on the quaternionic techniques introduced in the previous paper, and instead using an abstract formulation which can be investigated independently of coordinate system.  Furthermore, this approach leads to more general maps referred to as causal morphisms, and it will be shown that these do admit nontrivial solutions.

The idea of a causal morphism is to define a map that does not map points to points, but instead locally maps null surfaces to null surfaces. The idea is partially motivated by twistor theory, and its ambitwistor extension, where null lines and $\alpha$-planes are considered as fundamental objects of the theory\cite{ward1991twistor}\cite{dunajski2010solitons}.  In particular, conformal transformations on a self-dual 4-manifold induce holomorphic transformations on the corresponding twistor space\cite{10.2307/79638}.  This result directly motivates the idea of a self-dual morphism, and the extension to ambitwistor space results in the definition of causal morphisms.  

The maps defined here could be useful for a number of physical and mathematical applications.  A particularly interesting example is that of Yang-Mills theory, which is conformally invariant in four dimensions at the classical level.   It will be interesting to analyze the transformation properties for Yang-Mills fields under these maps, given their relationship to conformal transformations.  Furthermore, the inherent causal behavior of these maps suggests additional applications for physics, such as symmetries of general relativity and other physical theories, or generalized conformal structures on manifolds. Finally, it will be interesting to further understand the algebraic and analytic behavior of the solutions to these equations.  All of these applications will be left for future investigation.  

\section{Definitions and initial construction}

We start with a complex manifold $\mathcal{M}$ of dimension $n$ with a conformal structure that contains null tangent vectors. In this paper we will be interested primarily in $M=\mbb{C}^4$, which is also referred to as complexified Minkowski space, but some definitions can be made for more general manifolds. Consider a holomorphically embedded null surface $\chi(s):\mbb{C}\rightarrow \mathcal{M}$ which is non-singular, so that $\dot{\chi}$ is a non-zero null vector.  A causal morphism $f$ is a map on such embeddings, inducing a new embedding $f\circ \chi (s)$ which is also null.  We will demand that $f$ is holomorphic, in a sense to be defined, and that it is local.  In this context, we define locality by
\begin{mydef}
A map on embeddings $f$ is local if for any two holomorphic embeddings $\chi(s)$ and $\psi(t)$ that are tangent at points $s_0$ and $t_0$, the induced embeddings $f\circ \chi(s)$ and $f\circ \psi(t)$ are also tangent at $s_0$ and $t_0$. 
\end{mydef}  
The locality condition allows for a simpler characterization of the maps. To see this, we first define
\begin{mydef}
The ambitwistor correspondence space $(\mathcal{G},p,\mathcal{M})$ is the projective bundle of null tangent vectors over $\mathcal{M}$.  The fiber at a point $x\in \mathcal{M}$ is given by 
\begin{equation}
\mathcal{G}_x=\{v\in T_x \mathcal{M} | <v,v>=0, v\neq 0, v\sim \lambda v \ \forall \lambda \in \mbb{C}^*\},
\end{equation}
and $p$ is the natural projection operator.
\end{mydef}
 A nonsingular null embedding $\chi(s)$ induces an embedding $\chi_1(s):\mbb{C}\rightarrow \mathcal{G}$ given by $\chi_1(s)=(\chi(s),\dot{\chi}(s))$, where the second coordinate is defined up to scalar multiplication.  A map on embeddings $f$ can then be constructed from a map $f_1:\mathcal{G}\rightarrow \mathcal{G}$ by demanding that $f\circ \chi(s)=p\circ f_1\circ \chi_1(s)$.   For $f$ to define a causal morphism, $f_1$ must satisfy a consistency condition.  Writing $f_1\circ \chi_1(s)=(f\circ \chi(s), v\circ \chi_1(s))$, we require that
\begin{equation}
\lambda \frac{\partial}{\partial s} (f\circ \chi(s))= v\circ \chi_1(s), 
\end{equation} 
which must be satisfied for any $\chi(s)$, and some nonzero $\lambda$.  This will be referred to as the tangent consistency condition.  The locality condition guarantees that any causal morphism can be constructed from such an $f_1$.  Furthermore, the space $\mathcal{G}$ has a natural complex structure, so the holomorphicity of $f$ can be defined by the holomorphicity of $f_1$.  We can therefore equivalently define a causal morphism as a holomorphic map $f_1:\mathcal{G}\rightarrow \mathcal{G}$ satisfying the tangent consistency condition.  This is all summarized by the definition\footnote{This definition holds for all $\mathcal{G}$, but could be generalized by considering a subset $U\in\mathcal{G}$ as the domain.}
\begin{mydef}
\label{def:cmorph}
A causal morphism of $\mathcal{M}$ is a map $f$ on nonsingular null curves in $\mathcal{M}$.  This map is required to be local and holomorphically varying in the sense defined above.  The map can also be characterized by a holomorphic endomorphism $f_1:\mathcal{G}\rightarrow \mathcal{G}$ satisfying the tangent consistency condition. 
\end{mydef}

We will now focus our attention on the case where $M$ is complexified Minkowski space.  In this case, we can also consider morphisms of self-dual surfaces, which are defined to be holmorphically embedded two-surfaces with $\alpha$-plane tangents at each point.  In this case, locality can be defined similarly, with tangent lines being replaced with $\alpha$-planes.  Furthermore, we can define the twistor correspondence space $(\mathcal{F},p,M)$, whose fiber is the set of $\alpha$-planes passing through a point $x\in M$, parameterized by a projective spinor $\pi_{A'}$ at $x$.  The tangent consistency condition can also be readily generalized with tangent $\alpha$-planes. This leads us to the definition
\begin{mydef}
A self-dual morphism of $M$ is a map $f$ on nonsingular self-dual surfaces in $M$,  required to be local and holomorphically varying.  These maps can be characterized by holomorphic endomorphisms of the twistor correspondence space $f_1:\mathcal{F}\rightarrow\mathcal{F}$ satisfying the tangent consistency condition.
\end{mydef} 

It turns out that there is a relatively straightforward construction which can be used to find an infinite class of non-trivial examples of self-dual morphisms.  A further construction can then be used to give examples of causal morphisms in four complex dimensions.   

For the construction of self-dual morphisms, we will first review a correspondence introduced by Shaw\cite{Shaw_1986}, where holomorphically embedded null surfaces $\chi$ were shown to correspond uniquely to maps into the twistor space of $M$, which is $CP^3$.\footnote{More precisely the twistor space is $CP^3-CP^1$ for non-compact $M$.  For our purposes, however, this distinction will not be relevant, because the constructions are also valid for compactified $M$.}  This correspondence can be described geometrically by defining a map $\kappa$ which sends a null embedding $\chi:\mbb{C}\rightarrow M$ at a point $\chi(s)$ to the unique $\alpha$-plane $Z\in CP^3$ which passes through $\chi$ and contains the null vector $\dot{\chi}$ in its tangent space. This yields a new embedding $z:\mbb{C}\rightarrow CP^3$ defined by $z=\kappa\circ \chi$. The map $\kappa$ can also be inverted for a given curve $z(s)$ in a non-singular region to give a null curve $\chi=\kappa^{-1}\circ z$, so one can treat this as a correspondence.  The singular cases are studied in more depth by Shaw, and further details of this correspondence will be presented in the next section.  It should also be noted that the same construction can be made for dual-twistors, yielding similar examples for anti-self-dual (ASD) morphisms.

Now consider a holomorphic endomorphism $\tilde{f}:CP^3\rightarrow CP^3$.  These maps are well understood for any projective space $CP^k$, and are given by polynomial self-maps $F=(F_0,\cdots, F_{k})$ on $\mbb{C}^{k+1}$ such that $F^{-1}(0)={0}$ and the components $F_i$ are homogeneous polynomials of the same degree $d\geq 1$ (see e.g. \cite{Dinh2010}).  Given such a map, it could be possible to construct a causal morphism defined by $f \circ \chi=\kappa^{-1}\circ \tilde{f}\circ \kappa \circ \chi$.  This construction is summarized by the diagram
\begin{equation} \label{eq:fdef}
\hspace*{-2cm} 
\begin{tikzcd}
 & CP^3 \arrow[r, "\tilde{f}"] \arrow[d,leftarrow, "\kappa"]
& CP^3 \arrow[d, "\kappa^{-1}"] \\
\mbb{C}\arrow[r, "\chi"] & M \arrow[r, "f"]
& M
\end{tikzcd}
\end{equation}

This indeed defines a holomorphically varying map on null surfaces, because all of the maps involved are holomorphic.  The maps constructed in this way, however, do not necessarily satisfy the locality condition.  It will turn out, however, that this construction does yield self-dual morphisms, because the maps are local with respect to $\alpha$-plane tangents. We will prove this result in the next section, and then we will calculate some examples of self-dual morphisms generated in this way.
 
 \section{Self-dual morphisms}
 
To prove locality of the above construction, and thus provide examples of self-dual morphisms, we first review some relevant details of the Shaw correspondence.    Given a null embedding $\chi(s):\mbb{C}\rightarrow M$, the point $z(s)=\kappa\circ \chi(s)$ is given by the unique $\alpha$-plane satisfying
\begin{align}
\label{eq:omegaeq}
\begin{split}
\omega^A=i\chi^{AA'}\pi_{A'} , \ \ \  
\dot{\chi}^{AA'}\pi_{A'}=0, 
\end{split}
\end{align}
where $z^\alpha=(\omega^A,\pi_{A'})$ are homogeneous spinor coordinates for the point $z\in CP^3$, and the dependence on $s$ is suppressed.  The first equation says that $z$ contains $\chi$, and the second that it contains the null tangent $\dot{\chi}$, which can be written in the form 
\begin{equation}\label{eq:nulltangent}
\dot{\chi}^{AA'}=\lambda^A\pi^{A'}
\end{equation}
for some $\lambda^A$.  We can also differentiate equation \eqref{eq:omegaeq} to find
\begin{equation}\label{eq:zdoteq1}
\dot{\omega}^A=i\chi^{AA'}\dot{\pi}_{A'}.
\end{equation}

Conversely, given a non-singular map $z:\mbb{C}\rightarrow CP^3$, the above equations yield an inverse map $\chi=\kappa^{-1}\circ z$ given by 
\begin{equation}\label{eq:inversemap}
i\chi^{AA'}=(\omega^A\dot{\pi}^{A'}-\dot{\omega}^{A}\pi^{A'})/(\pi_{C'}\dot{\pi}^{C'}).
\end{equation}
To define a self-dual morphism, we consider a family of curves passing through a given point where $\lambda^A$ in equation \eqref{eq:nulltangent} is allowed to vary.  These curves are all tangent to the $\alpha$-plane defined by $z$.  The image of these curves under $f$ then defines a new family of curves, which define a new $\alpha$-plane in the target space, which gives

\begin{mythm}
The map $f\circ \chi$ defined by  diagram \eqref{eq:fdef} gives rise to a holomorphic endomorphism $f_1:\mathcal{F}\rightarrow \mathcal{F}$, and therefore defines a self-dual morphism.
\end{mythm}

\begin{proof}
Consider a null curve $\chi$ that is tangent to the $\alpha$-plane $z^\alpha=(\omega^A,\pi_{A'})$ at a point $\chi(s)$.  
We want to show that $\xi=f\circ \chi$ and its $\alpha$-tangent only depend on the point $(\chi,\pi)\in\mathcal{F}$ and not on any higher derivatives of $\chi$.  To fix notation, we also define the map $y=\tilde{f}\circ z$. Note that once the point $\xi$ is shown to be local, the $\alpha$-tangent is fully determined by the point $y$, so it will suffice to prove locality for $\xi$.  The inversion formula \eqref{eq:inversemap} shows that $\xi$ only depends on $(y,\dot{y})$, which in turn only depends on $(z,\dot{z})$.  So we will begin by maximally constraining $(z,\dot{z})$ using the point $(\chi,\pi)\in \mathcal{F}$.

We first note that $z$ is fully constrained by $(\chi,\pi)$, but equation \eqref{eq:zdoteq1} does not fully constrain $\dot{z}$.  However, it does constrain it to a two complex dimensional subspace of the full tangent space of $C^4$.  This subspace contains the solution $\dot{z}=r z$ for $r\in \mbb{C}^*$, which is not part of the tangent space to $CP^{3}$, as it leads to a rescaling. Because the map $\tilde{f}$ is homogeneous of degree $d$, this vector is mapped to a similar rescaling vector for $y$, which in turn leaves the inversion formula \eqref{eq:inversemap} invariant.  Any linear combination $(z,\dot{z}+r z)$ is then mapped to the same $\xi$.  Furthermore, any rescaling of $\dot{z}$ leads to a rescaling of $\dot{y}$, which also leaves the inversion formula invariant.  Therefore, any vector $\dot{z}$ in the two dimensional subspace defined by equation \eqref{eq:zdoteq1} that is not proportional to $z$ will lead to the same $\xi$, which therefore only depends on the point $(\chi,\pi)$.  The spinor $\tilde{\pi}$ defining the $\alpha$-plane tangent at $\xi$ is then determined directly by $y$.

We have thus shown that the map $\xi$ is local and constrained to an $\alpha$-plane $(\xi,\tilde{\pi})$ under the image of $f$.   Considering different curves tangent to $z$ at $\chi$ then yields a holomorphic endomorphism $f_1:\mathcal{F}\rightarrow \mathcal{F}$ given by $f_1(\chi,\pi)=(\xi,\tilde{\pi})$, and therefore defines a self-dual morphism.
\end{proof}

We will now calculate some examples of self-dual morphisms, with details given in appendix \ref{sec:AppA}. The simplest case is when $\tilde{f}$ is a projective transformation. In this case, there is a matrix $F$ that acts on $z$ to give $y=Fz$, with derivative $\dot{y}=F\dot{z}$.  Given $(\chi,\pi)\in\mathcal{F}$, the point $z=\kappa\circ \chi$ is determined by equation \eqref{eq:omegaeq}.  To find $\dot{z}$, we choose any spinor $\psi$ that is linearly independent from $\pi$, and use equation \eqref{eq:zdoteq1} to write $\dot{z}=(i\chi\psi,\psi)$. The choice of $\psi$ determines the vector in the two-dimensional subspace from the proof, which will not affect the final result.  Writing $F$ in block form, 
\begin{equation}\label{eq:Fblock}
F=
\begin{pmatrix}
A^A_{\ts B} & B^{AB'} \\
C_{A'B} & D_{A'}^{\ts B'} 
\end{pmatrix},
\end{equation}
equation \eqref{eq:inversemap} gives
\begin{equation}
\label{eq:genmob}
i\xi^{AA'}=(iA\chi+B)^{AB'}(iC\chi+D)^{-1 \ts A'}_{\ts \, B'},
\end{equation}
and $\tilde{\pi}$ is determined by $y=Fz$. In this case, $\xi$ is independent of $\pi$, and $f$ takes the form of a linear fractional transformation, or generalized M\"{o}bius transformation. This is encouraging, because it shows that these maps reduce to conformal transformations in the simplest case.  

In the case where $\tilde{f}$ is a degree two map, there are symmetric matrices $F$, $G$ such that $y=(z^T F z, z^T G z)$.  The derivative is then given by $\dot{y}=2(z^T F \dot{z}, z^T G \dot{z})$, and we again choose $\psi$ so that $\dot{z}=(i\chi \psi,\psi)$.  Writing the matrices in block form
\begin{equation*}
F^A=
\begin{pmatrix}[1]
[A^A]_{BC} & [B^{A}]_{B}^{\ts C'} \\
[B^{AT}]^{C'}_{\ts B} & [C^A]^{A'B'} 
\end{pmatrix}, 
\ G_{A'}=\begin{pmatrix}[1]
[D_{A'}]_{BC} & [E_{A'}]_{B}^{\ts C'} \\
[E^T_{A'}]^{C'}_{\ts B} & [F_{A'}]^{A'B'} 
\end{pmatrix}, 
\end{equation*}
we find 
\begin{equation}\label{eq:sdtwo}
i\xi^{AA'}=-M^{AB'}(N^{-1})^{\ts A'}_{B'},
\end{equation}
where we have defined the matrices
\begin{align}
\label{eq:mndef}
\begin{split}
&M^{AB'}=(\chi^{CC'}A^A_{\ts CB}\chi^{BB'}-i\chi^{CC'}B^{A\ts B'}_{\ts C}-iB^{A\ts C'}_{\ts B}\chi^{BB'}-C^{AC'B'})\pi_{C'}, \\
&N_{A'}^{\ts B'}=(\chi^{CC'}D_{A'CB}\chi^{BB'}-i\chi^{CC'}E_{A'C}^{\ts \ts B'}-iE^{\ts \ts C'}_{A'B}\chi^{BB'}-F_{A'}^{\ts C'B'})\pi_{C'}.
\end{split}
\end{align}
This expression has a form similar to equation \eqref{eq:genmob}, with quadratic factors instead of linear, and it does depend explicitly on $\pi$.  It is the simplest example of a self-dual morphism that is not a conformal transformation, and is useful for illustrating the general ideas behind the construction.  

\section{Causal morphisms}

As mentioned previously, the construction used for self-dual morphisms does not work for causal morphisms, because the resulting map does not necessarily satisfy the locality condition.  This is because the inversion formula for $\lambda$ from equation \eqref{eq:nulltangent} takes the form
\begin{equation}
 i\lambda^A=(i\chi^{AA'}\ddot{\pi}_{A'}-\ddot{\omega}^A)/(\pi_{C'}\dot{\pi}^{C'}),
\end{equation} 
depending explicitly on second derivatives of $z$.  Unlike the case with self-dual morphisms, these second derivatives are not adequately constrained by $(\chi,\dot{\chi})\in\mathcal{G}$ to define a map $f_1:\mathcal{G}\rightarrow \mathcal{G}$.

There is a way, however, to find examples of causal morphisms.  To do this, we first define the biquaternion projective space as follows
\begin{mydef}
The biquaternion projective space $BP^n$ is defined by a set of homogeneous coordinates $[b_0,\ldots,b_n]$ for $b_i\in M_2(\mbb{C})$ which are identified under right multiplication $[b_0 u,\ldots,b_n u]$ for $u\in GL_2(\mbb{C})$, and it is required that at least one of the $b_i$ is invertible.  
\end{mydef} 
This definition makes use of the fact that the biquaternions are isomorphic to $M_2(\mbb{C})$.  As with the quaternion projective space, this definition can equivalently be made using left multiplication, which is related by a change in orientation.  Some details and conventions used in this section are included in appendix \ref{sec:AppB}.

For our purposes we will mainly be interested in $BP^1$ with homogeneous coordinates $[b_0,b_1]$, which is the relevant space for four complex dimensions.  We will generally be interested in defining an $\alpha$-plane $z$ that passes through $b_0$, and a $\beta$-plane $w$ that passes through $b_1$.  In a local coordinate patch, this will define an $\alpha$-plane and a $\beta$-plane that intersect, and therefore defines a null line.  To see this explicitly, 
consider the coordinate patch $U_1$ where $b_1$ is invertible, and define the coordinate $x=b_0 b_{1}^{-1}$. The points on the $\alpha$-plane through $b_0$ can be written $b_z^{AA'}=b_0^{AA'}+\delta^A\pi^{A'}$, where $\delta$ is allowed to vary, and $b_1$ is held fixed.  In $U_1$, these points are given by 
\begin{equation}
x_z^{AA'}=x^{AA'}+\delta^A(\pi b_1^{-1})^{A'},
\end{equation}
which defines an $\alpha$-plane with coordinates $z_U=(\omega,\pi b_1^{-1})$ passing through $x$.  

Similarly, we can consider a $\beta$-plane with coordinates $w_{\alpha}=(\lambda_A,\mu^{A'})$ passing through $b_1$, which satisfies
\begin{align}
\label{eq:mueq}
\begin{split}
\mu^{A'}=-i\lambda_A b_w^{AA'}, \ \ \ 
\lambda_A \dot{b}_w^{AA'}=0.
\end{split}
\end{align}
The points on the $\beta$-plane can be written $b_w^{AA'}=b_1^{AA'}+ \lambda^A\delta^{A'}$.  In $U_1$, this plane becomes curved.  However, to first order in $\delta$ we find\footnote{See appendix \ref{sec:AppB} for some technical details of this calculation.}
\begin{equation}\label{eq:betapatch}
x_w^{AA'}=x^{AA'}+(x\lambda)^A(\delta b_1^{-1})^{A'},
\end{equation}
which defines a $\beta$-plane that is tangent to the curve at $x$, and has coordinates $w_U=(x\lambda, \mu)$.  The intersection of these $\alpha$ and $\beta$ planes in $U_1$ defines a null tangent vector 
\begin{equation}
v^{AA'}=(x\lambda)^A(\pi b_1^{-1})^{A'}
\end{equation} at $x$, and therefore a point $(x,v)$ in the correspondence space $\mathcal{G}$.  Naturally, a similar calculation can be done for the coordinate patch $U_0$, but the null tangent is defined independently of coordinates.

To further understand this construction, we must consider how these planes transform under the right action $(b_0,b_1)\rightarrow (b_0u,b_1u)$.  Under this transformation,  the $\alpha$ and $\beta$-planes defined above transform as 
\begin{equation}\label{eq:equivtrans}
(\omega,\pi)\rightarrow (\omega,u^{-1}\pi), \ \ \ (\lambda,\mu)\rightarrow (\lambda,\mu u).
\end{equation} 
These transformations are an extension of the equivalence relation $(b_0,b_1)\sim (b_0u,b_1u)$, and are necessary for the point in $\mathcal{G}$ to be well defined on the projective space.

This brings us to the construction of causal morphisms.  A causal morphism can be characterized by a holomorphic map $f_1:\mathcal{G}\rightarrow \mathcal{G}$.  Based on the above discussion, for the projective space this is equivalent to a map on $\alpha$-planes through $b_0$, and $\beta$-planes through $b_1$, which transform according to equation \eqref{eq:equivtrans} under the right action $(b_0,b_1)\rightarrow (b_0u,b_1u)$. We will therefore be interested in finding endomorphisms of $CP^3\times CP^3$ that preserve the desired equivalence relations.  Given such an endomorphism, diagram \eqref{eq:fdef} and its anti-self-dual analogue can be used to define self-dual and anti-self-dual morphisms on $b_0$ and $b_1$, which should lead to a causal morphism.  Because the inverse maps depend on the derivatives $\dot{z}$ and $\dot{w}$, however, one must also consider the transformation properties of the tangent spaces to make sure the maps are well defined.

In this paper we will not attempt to find the most general endomorphisms that satisfy the correct transformation properties.  Instead, we note that if the endomorphism is invariant under the right action, then this construction leads to a causal morphism.  In this case the map $\tilde{f}\circ \kappa \circ \chi$ from diagram \eqref{eq:fdef} and its anti-self-dual analogue do not depend on the equivalence classes chosen in the intermediate steps, and so the map must be well defined. Therefore, we will want to find endomorphisms of $CP^3\times CP^3$ that are invariant under the transformations \eqref{eq:equivtrans}. 

To construct such  endomorphisms, we first note that $\omega$, $\lambda$, and $\mu \cdot \pi$ are invariant under \eqref{eq:equivtrans}.  Therefore, any transformation constructed from these quantities will also be invariant. 
Furthermore, we require that any scale transformations on either $z$ or $w$ induce scale transformations on the target space, where $z=(\omega,\pi)$ and $w=(\lambda,\mu)$ are treated as vectors in $\mbb{C}^4$.  Writing the endomorphism as $(z,w)\rightarrow (\tilde{z},\tilde{w})$, this implies that components of $\tilde{z}$ have the same degree of homogeneity in both $z$ and $w$, and similarly for $\tilde{w}$.  Finally, analogous to the case with holomorphic endomorphisms of $CP^3$, we demand that the inverse image of any  projectively ill-defined sets, including $(0,\tilde{w})$ and $(\tilde{z},0)$ and points where both $\tilde{z}$ and $\tilde{w}$ are zero divisors, are also ill-defined. Endomorphisms constructed in this way will satisfy the correct transformation properties to define a causal morphism.   

For an example of this construction, consider the endomorphism defined by 
\begin{equation}
\begin{matrix}[1.3]
(\omega^A,\pi_{A'}) \\(\lambda_A,\mu^{A'})
\end{matrix} \
\rightarrow \
\begin{matrix}[1.3]
(A^A_{\ts BC}\,\omega^B\lambda^C+B^A\, \mu\cdot \pi,  \ C_{A'BC}\,\omega^B\lambda^C+D_{A'}\, \mu\cdot \pi)\\
(E_{ABC}\,\omega^B\lambda^C+F_A\, \mu\cdot \pi,  \ G^{A'}_{\ts BC}\,\omega^B\lambda^C+H^{A'}\, \mu\cdot \pi).
\end{matrix}
\end{equation}
Using the procedure for constructing self-dual morphisms outlined in the previous section, and its anti-self-dual analogue, this gives rise to the map
\begin{equation}
\begin{matrix}[1.3](b_0^{AA'},\pi_{A'}) \\(\lambda_{A}, b_1^{AA'})\end{matrix} \ 
\rightarrow
 \ \begin{matrix}[1.3]
 ((iA\lambda b_0+B\mu )^{AB'}(iC\lambda b_0+D\mu)^{-1\ts A'}_{\ts B'}, \, C_{A'BC}\,\omega^B\lambda^C+D_{A'}\, \mu\cdot \pi) \\
(E_{ABC}\,\omega^B\lambda^C+F_A\, \mu\cdot \pi, \, -(E\omega -i F\pi b_1 )^{-1 \ts A}_{\ts C}(G\omega-H\pi b_1)^{A'C} ),
\end{matrix}
\end{equation}
where natural contractions are implied.
Based on the discussion above, this gives rise to a causal morphism.  Restricting to the coordinate patch $U_1$ for both spaces yields a map $f_1:\mathcal{G}\rightarrow \mathcal{G}$ as in definition \ref{def:cmorph}.

\section{Conclusion}

In this paper we have defined a class of transformations which map null surfaces to null surfaces in four complex dimensions, and which contain the group of conformal transformations as a special case.  These transformations are more conveniently characterized by mappings on the null tangent bundle, and preserve the causal structure of the space, even though they do not preserve points.  The natural relationship to causality and conformal structures makes these maps potentially useful for both mathematical and physical applications.

We have constructed examples of these transformations by using the correspondence between null surfaces in $CM$ and holomorphic maps into twistor space investigated by Shaw\cite{Shaw_1986}.  This led to explicit expressions for self-dual morphisms involving degree one and degree two endomorphisms of twistor space.  The degree one endomorphisms are  generalized M\"{o}bius transformations, and the degree two endomorphisms are quadratic rational functions.   A further construction based on biquaternion projective space was then introduced to generate examples of causal morphisms as well.  

It will be interesting to further understand the properties of these transformations, and also consider their potential applications.  The preservation of causal structure make them a natural candidate for a number of physical applications, including applications to Yang-Mills theory and general relativity.  The relationship to conformal transformations also makes them interesting from a purely mathematical point of view, particularly for the study of four dimensional manifolds.  We look forward to further understanding the properties and applications of these transformations in future work.

\appendix

\section{Details of self-dual morphisms} \label{sec:AppA}

Here we present a more detailed calculation for the examples of self-dual morphisms.  It was shown that when $\tilde{f}$ is a projective transformation, we have $\tilde{f}_1(z,\dot{z})=(Fz,F\dot{z})$.  
Using the block form of F given by equation \eqref{eq:Fblock} and the inversion formula \eqref{eq:inversemap}, we find
\begin{equation}\label{eq:yfirst}
i\xi^{AA'}=\dfrac{2\, G^{AB'}H^{A'C'}\pi_{[B'}\psi_{C']}}{\varepsilon^{D'E'}H_{D'}^{\ts F'}H_{E'}^{\ts G'}\pi_{F'}\psi_{G'}}.
\end{equation}
where we have defined the matrices
\begin{align*}
G^{AB'}=iA_{\ts B}^A\chi^{BB'}+B^{AB'}, \ \ H^{A'B'}=iC^{A'}_{\ts B}\chi^{BB'}+D^{A'B'}.
\end{align*}
Here, the numerator can be simplified by using the identity
\begin{equation}\label{eq:antisymident}
2\pi_{[B'}\psi_{C']}=\pi_{D'}\psi^{D'}\varepsilon_{B'C'},
\end{equation}
which is true for any spinors $\pi$ and $\psi$.  Furthermore, we note that the denominator is invariant under antisymmetrizing indices $F'$ and $G'$.   Antisymmetrizing and using equation \eqref{eq:antisymident}, the denominator can then be written $\det(H^{C'}_{\ts D'})\pi_{B'}\psi^{B'}$, where we have also used the identity 
\begin{equation*}
2\det(H^{B'}_{\ts D'})=-\varepsilon_{B'C'}\varepsilon^{D'E'}H^{B'}_{\ts D'}H^{C'}_{\ts E'}.
\end{equation*}
We can then use the identity $(H^{-1})_{A'}^{\ts B'}=-H_{\ts A'}^{B'}/\det(H^{C'}_{\ts D'})$ to find
\begin{equation}
i\xi^{AA'}=G^{AB'}(H^{-1})_{B'}^{\ts A'},
\end{equation}
which is equation \eqref{eq:genmob}.

For the second order case, the expression obtained from a direct use of equation \eqref{eq:inversemap} is
\begin{equation}
i\xi^{AA'}=\dfrac{2M^{AB'}N^{A'C'}\pi_{[B'}\psi_{C']}}{\varepsilon^{D'E'}N_{D'}^{\ts F'}N_{E'}^{\ts G'}\pi_{F'}\psi_{G'}},
\end{equation}
where the matrices are defined by equation \eqref{eq:mndef}.  This equation has the same form as equation \eqref{eq:yfirst}, so it can be simplified in the same way, yielding equation \eqref{eq:sdtwo}.

\section{Details of biquaternion projective space}
 \label{sec:AppB}
 
 Here we will give some details about the construction of biquaternion projective space.  With indices, the space is defined by identifying the homogeneous points $[(b_0)^{AA'},...,(b_n)^{AA'}]$ with $[(b_0)^{AB'}u_{B'}^{\ts A'},...,(b_n)^{AB'}u_{B'}^{\ts A'}]$ for any invertible $u$.  In a coordinate patch $U_i$, we'd like to choose $u=b_i^{-1}$.  There is a technical difficulty with this choice, because $u$ has two primed indices, while $b_i$ has one primed and one unprimed index, which is a different type of matrix.  Stated differently, the formula $b_i^{AB'}(b_i^{-1})_{B'}^{\ts A'}=I^{AA'}$ does not really work, because the isotropic tensor $\delta^{A}_{B}$ is usually defined with two of the same indices (primed or unprimed), and has one index raised and one lowered.  The resolution to this issue is that we do not view the identity matrix $I^{AA'}$ used here as special, and in fact any invertible matrix could be used.  One can readily check that all the calculations would be equally valid if a different matrix was chosen, it would just define a coordinate patch that is easily related to the choice made using $I^{AA'}$.  For convenience we still use the notation $b^{-1}$ for this matrix, even though the choice is not invariant under coordinate transformations.
 
Some of the technicalities regarding indices come into play with the derivation of equation \eqref{eq:betapatch}.  This equation is derived by taking the derivative of $x=b_0 b_1^{-1}$ for $b_0$ fixed and $b_1$ constrained to the $\beta$-plane $w$, and using the identity
\begin{equation}
d(b^{-1})_{B'}^{\ts A'}=-(b^{-1})_{B'A}\, d(b^{AC'})(b^{-1})_{C'}^{\ts A'}.
\end{equation} 
It is clear from the indices here that the left inverse is not the same as the right inverse.  With explicit indices, equation \eqref{eq:betapatch} is \begin{equation}
x_w^{AA'}=x^{AA'}+b_0^{AB'}(b_1^{-1})_{B'}^{\ts B}\lambda_B\delta^{C'}(b_1^{-1})_{C'}^{\ts A'},
\end{equation}
 where we have used the spinor ``see-saw'' on the $B$ index.  Numerically, the inverse on the left is the same as the one with two primed indices, but could change under a coordinate transformation.  For notational clarity we have suppressed indices and ignored this subtlety, but it is useful to be aware of for detailed calculations.

\bibliographystyle{Science}
\bibliography{References}

\end{document}